\documentclass[sumlimits,intlimits,namelimits]{amsart}

\usepackage{bm}

\usepackage{amssymb,latexsym}
\usepackage[numeric,initials,bibtex-style,nobysame]{amsrefs}
\usepackage{amsthm}

\usepackage{a4wide}
\usepackage[american]{babel}

\usepackage{eucal}
\usepackage{mathrsfs}

\usepackage[dvipsnames]{xcolor}

\def\pg{\mathhexbox278}

\newcommand{\Cinf}{\ensuremath{\mathcal{C}^\infty}}

\newcommand{\D}{\ensuremath{\mathcal{D}}}
\newcommand{\G}{\ensuremath{\mathcal{G}}}
\renewcommand{\S}{\mathscr{S}}

\newcommand{\mb}[1]{\ensuremath{\mathbb{#1}}}
\newcommand{\N}{\mb{N}}

\newcommand{\R}{\mb{R}}
\newcommand{\C}{\mb{C}}

\renewcommand{\d}{\ensuremath{\partial}}
\newcommand{\diff}[1]{\frac{d}{d#1}}

\newfont{\bl}{msbm10 scaled \magstep2}

\newtheorem{theorem}{Theorem}[section]
\newtheorem{lemma}[theorem]{Lemma}
\newtheorem{proposition}[theorem]{Proposition}

\newtheorem{corollary}[theorem]{Corollary}

\theoremstyle{definition}
\newtheorem{remark}[theorem]{Remark}

\newcommand{\beq}{\begin{equation}}
\newcommand{\eeq}{\end{equation}}

\newcommand{\isom}{\cong}
\newcommand{\col}{\colon}

\newcommand{\FT}[1]{\widehat{#1}}
\newcommand{\F}{\ensuremath{{\mathcal F}}}

\newcommand{\dis}[2]{\langle #1 , #2 \rangle}
\newcommand{\inp}[2]{\langle #1 | #2 \rangle}  
\newcommand{\notmid}{\mid\kern-0.5em\not\kern0.5em}

\newcommand{\norm}[2]{{\left\| #1 \right\|}_{#2}}

\newcommand{\al}{\alpha}
\newcommand{\be}{\beta}
\newcommand{\ga}{\gamma}

\newcommand{\de}{\delta}
\newcommand{\eps}{\varepsilon}

\newcommand{\vphi}{\varphi}

\newcommand{\la}{\lambda}

\newcommand{\om}{\omega}
\newcommand{\Om}{\Omega}
\newcommand{\sig}{\sigma}

\newcommand{\supp}{\mathop{\mathrm{supp}}}

\renewcommand{\Im}{\ensuremath{\mathop{\mathrm{Im}}}}

\newcommand{\ovl}[1]{\overline{#1}}

\newcommand{\gC}{\widetilde{\C}}

\newcommand{\tH}{\widetilde{H^2}}
\newcommand{\tL}{\widetilde{L^2}}

\usepackage{tikz-cd}
\usepackage{tikz}

\begin{document}

\pagestyle{plain}

\title{Modeling Schr\"odinger equation diffraction with generalized function potentials and initial values}

\author{G\"unther H\"ormann}

\author{Ljubica Oparnica}

\author{Christian Spreitzer}

\address{Fakult\"at f\"ur Mathematik\\
Universit\"at Wien, Austria}

\email{guenther.hoermann@univie.ac.at}

\address{Faculty of Education\\
University of Novi Sad, Serbia}

\email{ljubica.oparnica@gmail.com, ljubica.oparnica@pef.uns.ac.rs}

\address{Fakult\"at f\"ur Mathematik
\& Fakult\"at f\"ur Physik\\ Universit\"at Wien, Austria}

\email{christian.spreitzer@univie.ac.at}



\date{\today}

\begin{abstract} 
We discuss spectral properties of a regularization approach to a Schr\"odinger equation set-up for the diffraction of a quantum particle at almost planar patterns. Physically meaningful initial values and potentials are modeled in terms of regularizing families and the solutions can be interpreted as generalized functions. We establish spectral and scattering theoretical properties of the regularizing solution families and provide some comparison with the more direct approximations and simplifications used in physics.
\end{abstract}

\maketitle

%
%
%

%
%
%
%
%



\section{Introduction}

In the introductory parts of quantum physics text books, one often finds discussions of double-slit experiments. They perfectly serve as an illustration of the fundamental so-called quantum logic and are still a prominent subject of research, not only as a \emph{Gedankenexperiment} but also in contemporary experimental physics (see, e.g., \cite{Fein:19} and \cite{Brand:21}). The methodological parts of standard text book models are then based on Schr\"odinger equations with various potentials. 
However, as a somewhat surprising observation, it seems difficult to find  discussions of the foundational double-slit experiment among the quantum mechanical text book models with such Schr\"odinger equation set-up. The desire to have such a theory in early quantum physics is nicely illustrated by the following quote from Weinberg's text book (\cite[pages 14-15]{Weinberg:15}):``There is a story that in his oral thesis examination, de Broglie was asked what other evidence might be found for a wave theory of the electron, and he suggested that perhaps diffraction phenomena might be observed in the scattering of electrons by crystals. .... What was needed was some way of extending the wave idea from free particles, described by waves ..., to particles moving in a potential, ...''. 

The standard physics argument justifying the occurrence of the interference pattern observed in a double-slit experiment is, in fact, a pure classical approximation based on point sources for waves at the narrow slits.  The two discussions of the double-slit in \cite{Manoukian:89}  and \cite{Beau:12} are based on the Feynman path technique, which means that mathematically speaking they make use of the typical distribution theoretic fundamental solution to the free Schr\"odinger equation and produce a somewhat refined justification of the usual approximation via classical wave theory and diffraction. A recent numerical and qualitative analysis based on the Bohmian point of view with particle trajectories is presented in \cite{DDS:22}, which is very similar to our approach in Subsections \ref{regpot} and \ref{regin} in its modeling of the initial value and the potential. In fact, as noted already in \cite[Remark 2.4]{Hoermann:17}, the regularization methods from theories of generalized functions would apply also in studies of the Bohmian flow related to diffraction problems based on Schr\"odinger equations.

The double slit problem is often treated as a ``boundary value'' problem with a free wave function ``approaching the double-slit''. However, the Schr\"odinger equation does not have the property of finite propagation speed that a wave equation---or, more generally speaking, any strictly hyperbolic equation has. Thus on a fundamental level it is not justified a priori why an approximation with classical wave propagation starting from the slits works so well in explaining the measured diffraction pattern on a screen at some distance from the slits. Moreover, at least in principle and a priori, one cannot rule out partial reflections from the blocking objects and that a superposition of what came in from one side of the slits and what was reflected could generate a stationary state (as described in \cite[Subsection 3.7.3]{Hall:13}) corresponding to a solution of the time-independent Schr\"odinger equation (i.e., an eigenfunction of the Hamiltonian).

A further difficulty is that modeling the potential for a realistic double-slit configuration in a Schr\"odinger operator is mathematically delicate and even distributional potentials, e.g.\ producing Dirac-$\delta$-type ``sources for passing waves'' at the slits, seem not truly appropriate due to their idealizations from the outset. Instead one might rather strive for an implementation via regularizations that can more accurately capture the nature of an essentially insurmountable high barrier away from the slits that may at the same time be of infinitesimal width, thus almost located in a plane perpendicular to the ``main propagation direction'' as seen from the source. Thus we will attempt to accurately model such potentials by corresponding generalized functions which can be conveniently represented via families of regularizations. In addition, a realistic initial value configuration will not be modeled by some $L^2$ function, but rather by a family of wave packet type regularizations defining a generalized function. We emphasize that neither for the initial value nor for the potential are these regularizations expected to be convergent as distributions. Instead they only obey certain asymptotic estimates. For this reason, our approach cannot make direct use of the elaborated and successful Hilbert space theory of Schr\"odinger operators with delta-type singularities in the potentials as described in \cite{BEKS:94,BLL:13,ER:16,BEHL:17,CFP:19,CFPCorr:19,CFP:21}. However, the method of modeling impenetrable barriers as regularizations of infinite potentials and employing concepts from the theory of Colombeau generalized functions to prove existence and uniqueness of solutions and analyze their properties has a much wider range of applicability. For instance, as an alternative to boundary conditions, generalized infinite potentials may also be used to model geometric constraints, e.g. particles confined to a half-space or to a box.

In the context of the regularization based approach of Colombeau-generalized solutions, 
well-posedness of Cauchy problems for Schr\"odinger equations allowing for discontinuous or distributional coefficients, initial data, and right-hand sides was shown in  \cite{Hoermann:11}. Previously, several Colombeau-generalized solutions to special types of linear and nonlinear Schr\"odinger equations have been constructed in \cite{Bu:96,DN:19,Sto:06b,Sto:06a}. The special case of Schr\"odinger operators with $\delta$-potential could be treated in a nonstandard analysis setting (see \cite{Albeverio:88}), but also with quadratic forms in terms of a Friedrichs extension (see \cite[Example 2.5.19]{Thirring:02}). Further applications with a mixed setting involving distribution theory, Hilbert space techniques or measures and invariant means can be found in the context of seismic wave propagation (see \cite{dHHO:08}) or generalizations of the usual quantum mechanical initial values (see \cite{Hoermann:17}).   

Recent results in \cite{RST:23} also use regularization methods for Schr\"odinger-type partial differential operators with the standard spatial Laplacian replaced by a Laplace-Beltrami-type operator that is allowed to have a singular $x$-dependent coefficient. Employing concepts based on so-called very weak solutions which are obtained as families of regularizations, the main results are well-posedness in the corresponding sense even with very singular initial data and consistency with Sobolev space solutions in case the coefficient has better regularity. In a similar set-up, \cite{ARST:21} investigates Schr\"odinger operators of fractional differentiation order in space and with highly singular potentials. 

Recall that in quantum mechanics one is often interested in allowing for the multiplication operator term $V$ in the  Schr\"odinger equation $\d_t u = i \Delta_x u - i V u$ to model nonsmooth potentials, as, e.g., already with Coulomb-type potentials. In the classical $L^2$ theory we have initial data $u |_{t = 0} = u_0$ and $|u_0|^2$ gives an initial probability density while $|u(.,t)|^2$ represents the result of the evolution of this probability density at time $t > 0$.  

In the appendix we briefly review concepts from \cite{Hoermann:11} in the regularization approach to generalized functions in the sense of Colombeau. We also describe the main results (Theorem \ref{exunthm} and Corollary \ref{cor}) on unique existence of solutions to the following Schr\"odinger equation Cauchy problem: Find a (unique) generalized function $u$ on $\R^n \times [0,T]$ solving 
\begin{align}
  \d_t u - i\, \sum_{k=1}^n   \d_{x_k} (c_k 
    \d_{x_k} u) 
    + i V u &= f \label{SCPDE}\\
    u \mid_{t=0} &= g, \label{SCIC} 
\end{align}
where $c_1, \ldots, c_n$, $V$, and $f$ are generalized functions on $\R^n \times [0,T]$ and $g$ is a generalized function on $\R^n$. 
The appendix also indicates compatibility with classical and distributional solution concepts.

\subsection*{Outline of the paper and brief summary of the main results}

Section 2 discusses the details of the mathematical regularization modeling for the potential and the initial values, while Section 3 establishes the key spectral properties of these regularizations and ensures the applicability of basic methods from scattering theory. Section 4 makes an effort to relate the regularization and generalized function set-up  with various approximations or calculational simplifications employed successfully in physics.

The various results of our analysis can be summarized under the following aspects:
\begin{trivlist}

\item\emph{Modeling with regularizations and unique generalized solutions:} Subsections \ref{regpot} and \ref{regin} give precise specifications of the potential $V$ for scattering at a planar pattern, e.g., in a double-slit setting, and of the initial value $g$ in the form of a wave packet as regularizing families $(V_\eps)_{\eps \in (0,1]}$ and $(g_\eps)_{\eps \in (0,1]}$. 
It is guaranteed then by Theorem \ref{basicthm} that with a such a set-up we have uniqueness of Coulombeau-type generalized function solutions to the Schr\"odinger equation Cauchy problem   
$$
  \d_t u =  i\, \Delta u  - i V u, \quad  u|_{t=0} = g. 
$$

\item\emph{Spectral properties:} Subsection \ref{specprop} gives a careful analysis of the spectrum of the Hamilton  operator $H_\eps = - \Delta + V_\eps$ arising in these Coulombeau-generalized Schr\"odinger operators. In particular, the absence of bound states is shown by Theorem \ref{specthm} and that the essential spectrum is preserved from the free part $-\Delta$. Interpreting and extending then in Subsection \ref{geneig} the notion of eigenvalues to the context of Colombeau-type generalized functions, we determine in Proposition \ref{gspecprop} the generalized point spectrum and obtain that it contains the essential spectrum of $H_\eps$ as a subset.

\item\emph{Scattering theoretical representations:} The modeling options for the potential regularizations are so flexible that they allow also to have compact support of $V_\eps$ at each instance of $\eps > 0$ and letting the size of the supports grow as $\eps \to 0$. An advantage of such a set-up is that methods from the theory of short range perturbations become applicable, which we justify in Subsection \ref{scattsubs} and obtain a concrete distributional representation for the so-called distorted Fourier transform of the generalized solution in Equation \eqref{disform}.

\item\emph{Comparison with standard approximations in physics:} A common simplification of the model is to drop the potential from the operator and instead move information about some anticipated behavior in the scattering process or an intuitive qualitative property of the solution into the initial value or a new source term for the Cauchy problem. We study such strategies on a very basic level, but arrive already at the end of Subsection \ref{thphresults} at an intuitively appealing agreement with physics in terms of a wave intensity distribution on a parallel screen based on the Fourier transform of the characteristic function of the slit configuration. In Subsections \ref{apc} and \ref{iapc} we investigate the plausibility of approaches in physics that consider first the propagation of a free solution with initial values approximating a Dirac delta towards the scattering plane and then use essentially the interaction $V_\eps \cdot u_\eps$ of the potential function with this free solution at the scattering plane as a source term $F_\eps$ for a new Cauchy problem. The analysis is inherently vague at a few points, but we do obtain a certain coherence for the crucial Equation \eqref{wApprox} giving an approximate solution by two different methods. Moreover, we find in Proposition \ref{betaprop} an improved and clear-cut mathematical support for the precise choice of the source term $F_\eps$.

\end{trivlist}


\section{Modeling of the Cauchy problem and of the Hamiltonian}

\subsection{Regularizations representing the potential}\label{regpot} \hphantom{B}

\noindent
\parbox[c]{0.7\textwidth}{
We consider a model potential of the form $V(x,y,z) = b(x) h(y)$ for diffraction at a pattern concentrated near the plane $x = 0$, which is invariant with respect to height $z$ and has a horizontal structure described by $h$ as a generalized function of $y$. The value $h(y)$ should be ``essentially zero'' where slits are located, say if $y$ belongs to a subset $S \subseteq \R$ being the disjoint union of finitely many intervals, and $h(y)$ should be ``essentially infinite'' at locations that block classical objects, i.e., if $y \in \R \setminus S$. {\small (Typically, the intervals in $S$ will be bounded in case of modeling slits, but we need not require this here formally.)}
Since the whole problem is invariant with respect to $z$-translations, we reduce it immediately to a problem in the $(x,y)$-plane. 
}
\hspace{0.02\textwidth}
\parbox{0.25\textwidth}
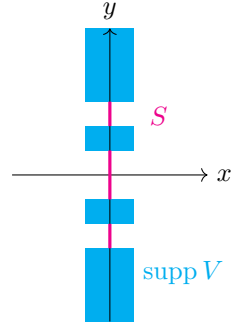
{
\begin{tikzpicture}[scale=0.65]
   \fill[cyan, opacity=0.2] (-0.5,-3) -- (-0.5,-1.5) -- (0.5,-1.5) -- (0.5,-3) -- cycle;
   \fill[cyan, opacity=0.2] (-0.5,-1) -- (-0.5,-.5) -- (0.5,-0.5) -- (0.5,-1) -- cycle;
   \fill[cyan, opacity=0.2] (-0.5,0.5) -- (-0.5,1) -- (0.5,1) -- (0.5,0.5) -- cycle;
   \fill[cyan, opacity=0.2] (-0.5,1.5) -- (-0.5,3) -- (0.5,3) -- (0.5,1.5) -- cycle;
   \draw[->] (-2,0) -- (2,0) node[right] {$x$};
   \draw[->] (0,-3) -- (0,3) node[above] {$y$};
    \draw[very thick, magenta] (0,-1.5) to (0,-1);
     \draw[very thick, magenta] (0,-0.5) to (0,0.5);
     \draw[very thick, magenta] (0,1) to (0,1.5);
     \node[magenta] at (1,1.2) {$S$};
     \node[cyan] at (1.5,-2) {$\supp V$};
\end{tikzpicture}
}

\medskip

As a representation of the potential we consider a  family of regularizations given by
\beq\label{potentialreg}
V_\eps(x,y) := b_\eps(x) h_\eps(y)   \quad (\eps > 0, (x,y) \in \R^2),
\eeq
where $b_\eps, h_\eps \in L^\infty(\R)$ are nonegative for every $\eps \in\, ]0,1]$, for some $c > 0$ we have $\supp b_\eps \subseteq [-c,c]$ for every $\eps$, and   
\beq\label{hlimits}
   \text{$h_\eps \to 0$ pointwise on  $S$, but $h_\eps \to \infty$ on $\R \setminus S$}. 
\eeq
In addition, we assume that the families $(b_\eps)$ and $(h_\eps)$ are $W^{\infty,\infty}$-moderate (cf.\ the appendix for this notion). We may then consider $V$ as the element in $\G_{\infty}(\R^2)$ represented by the family $(V_\eps)$ and apply Theorem \ref{exunthm}, given in the appendix, for arbitrary $T > 0$ to obtain the statement of unique solvability in the following theorem. 
 (Note that the conditions (a) and (b) in Theorem \ref{exunthm} are automatically satisfied, since we have $c_k = 1$ and a potential $V$ without $t$-dependence.)

\begin{theorem}\label{basicthm}
Let $V \in \G_{\infty}(\R^2)$ denote the potential defined via the regularizations \eqref{potentialreg}. For every $g \in \G_{2}(\R^2)$, there is a unique solution $u \in \G_{2}(\R^2 \times [0,T])$ to the Cauchy problem
\begin{equation}\label{ourproblem}
  \d_t u =  i\, \Delta u  - i V u, \quad  u|_{t=0} = g. 
\end{equation}
\end{theorem}

\begin{remark}(Relation with the indicator function of the slit configuration)\label{remindicator} The standard approximations in theoretical physics as described in Subsection \ref{thphresults} circumvent the introduction of an explicit potential for the Schr\"odinger operator. Instead they encode the slit configuration, corresponding to the subset $S \subseteq \R$ above, into a specific initial condition, often called a ``boundary condition'', involving the characteristic function $1_S$ of $S$, i.e., $1_S(y) = 1$ for $y \in S$ and $0$ otherwise. Thus, one would think of  $V=0$ and $g (x,y) = \de_0(x) \cdot 1_S(y)$ in Equation \eqref{ourproblem}. With a view at Equation \eqref{potentialreg}, a heuristics for the precise form of this initial condition includes the assumption $b_\eps \to \de_0$ as $\eps \to 0$, which is compatible with the specifications, and it remains to find some argument for replacing $h_\eps$ by $1_S$.

We may approximately justify this in our set-up: By construction, $\norm{h_\eps}{L^\infty} \to \infty$ as $\eps \to 0$, in particular,  $\norm{h_\eps}{L^\infty} > 0$ for small $\eps > 0$, and we clearly have $h_\eps(y)/\norm{h_\eps}{L^\infty} \to 0$ for any $y \in S$;  to guarantee also $h_\eps(y)/\norm{h_\eps}{L^\infty} \to 1$, if $y \in \R \setminus S$, so that we obtain the pointwise convergence
\beq\label{convind}
   1 - \frac{h_\eps(y)}{\norm{h_\eps}{L^\infty}} \to 1_S(y) \quad \eps \to 0,
\eeq
we could simply make the aspect $h_\eps \to \infty$ on $\R \setminus S$ of requirement \eqref{hlimits} more concrete, namely,   
\beq\label{uninfty}
   \norm{h_\eps}{L^\infty} \stackrel{\eps\to 0}{\longrightarrow} \infty  \quad\mathrm{and}\quad\forall y \in \R \setminus S\, \exists c > 0\, \exists \eps_0 > 0\col   h_\eps(y) \geq \norm{h_\eps}{L^\infty} - c > 0 \quad (0 < \eps \leq \eps_0).
\eeq
This condition would even be compatible with choosing $h_\eps$ of compact support (which is certainly growing as $\eps$ becomes smaller) and implies \eqref{convind}.
\end{remark}

In the description of a scattering experiment, one is interested in the limiting behavior as $\eps \to 0$ of solution representatives $(u_\eps)$ in case of an initial value $g$, modeled by a regularizing family $(g_\eps)$, corresponding to a quantum particle that is ``spread out'' considerably in the  $y$-direction  and  ``approaches'' the plane $x= 0$ from the region $x < 0$. 
In particular, the aim would be to study the properties of the \emph{intensity distribution} 
\beq\label{intensity}
    y \mapsto |u_\eps(t_1, x_1,y)|^2
\eeq
at some time $t _1> 0$ and as $\eps \to 0$, where $x_1 > 0$ represents the location of some screen parallel to the diffraction plane.

 For the purpose of asymptotics with $\eps \to 0$ one does not rely on the full framework of generalized functions and Theorem \ref{exunthm}, but may instead look more directly at the family of regularized problems
\begin{equation}\label{ourprobleps}
  \d_t u_\eps = i \, \Delta u_\eps  - i V_\eps u_\eps, \quad  u_\eps|_{t=0} = g_\eps. 
\end{equation}
For every $\eps > 0$, $V_\eps$ is bounded and real-valued, thus the \emph{Hamiltonians} 
\beq\label{Hamiltonian}
H_\eps := -\Delta + V_\eps 
\eeq
are an $\eps$-parametrized family of self-adjoint operators with domain $H^2(\R^2)$ (independent of $\eps$). We obtain the solution $u_\eps$ via the  unitary group generated by $H_\eps$, i.e.,
$$
  \forall \eps \in (0,1], \forall t \in \R \col \quad u_\eps(t) = e^{- i t H_\eps} g_\eps.
$$

\subsection{Regularizations representing the initial configuration}\label{regin}

Intuitively, during the time while the ``source quantum particle approaches the scattering obstacle'', the idealization of ``least possible localization in $y$''  would be a plane wave of the form $(x,y,t) \mapsto Ae ^{i (k x - \omega t)}$ ($A\in \mathbb C$), following de Broglie's correspondence with a free particle of momentum $\hbar (k,0) \in \R^2$ and energy $E = \hbar \om \in  (0,\infty)$, where $\omega =\frac{\hbar k^2}{2m}$ for a particle of mass $m$. However, it is a more realistic model to implement instead the initial value $g$, or rather its regularizations $g_\eps$, as a ``wave packet'' with an average momentum $p_0 = \hbar k$ in the $x$-direction. Dropping again the explicit reference to $\hbar$, we may make the ansatz for the initial value representatives as the family of functions
\beq\label{invalmod}
     g_\eps (x,y) =   \rho_{\eps}(y)  \vphi_\eps(x) e^{i p_0 x},
\eeq
where $(\rho_\eps)_{\eps \in (0,1]}$ and  $(\vphi_\eps)_{\eps \in (0,1]}$ are $H^\infty$-moderate families of nonnegative functions on $\R$ such that $\supp(\vphi_\eps) \subseteq (-\infty,-c-1]$, and $\norm{\rho_\eps}{L^2} = 1 = \norm{\vphi_\eps}{L^2}$ for every $\eps > 0$. Note that $H^\infty(\R) \subseteq H^1(\R) \subseteq \{ \psi \in C(\R) \mid \lim_{|z| \to \infty} |f(z)| = 0\}$, hence $\int_\R \rho_\eps(y) \rho'_\eps(y) \, dy = (\rho_\eps(y)^2/2) |^\infty_{-\infty} = 0$ and similarly for $\vphi_\eps$ and we obtain, 
by construction, that $g_\eps$ has the following momentum expectation values 
$$
      \inp{g_\eps}{-i \d_y g_\eps}  = 0  \quad\text{and}\quad  \inp{g_\eps}{-i \d_x g_\eps} = p_0.
$$
Moreover, $\norm{g_\eps}{L^2}=1$ and $\supp(g_\eps)\cap \supp(V_\eps)=\emptyset$ for all $\eps \in (0,1]$  (recall that $\supp(V_\eps)\subseteq [-c,c]\times \mathbb R$). Thus, the family of wave functions $(g_\eps)$   models a single quantum particle initially  located in a region that does not intersect with the set where the potential is nonzero. The particle moves with  velocity $\tfrac{p_0}{m}$ towards the double-slit and as the wave packet propagates, it spreads according to the dispersion relation $\omega(k)=\tfrac{\hbar k^2}{2m}$. For example, if $g_\eps$ is a Gaussian, then the corresponding solution to the free equation will also be a Gaussian, but its width will increase  with time $t$. Note that if $\varphi_\eps$ has compact support (which would appear as a rather natural assumption for a particle emanating from some source), its Fourier transform cannot be compactly supported  by the Paley-Wiener-Schwartz theorem. Hence the wave packet will contain arbitrarily large frequencies associated to its direction of propagation. This implies that for all $t>0$, $\supp(u_\eps(t))\cap \supp(h_\eps)$ will be nonempty.  

\section{Analysis of the Hamiltonian}

We study now for arbitrary $\eps > 0$ the self-adjoint operator $H_\eps = - \Delta + V_\eps$ with domain $H^2(\R^2)$ according to \eqref{Hamiltonian}, where $V_\eps (x,y) = b_\eps(x) h_\eps(y)$ as in \eqref{potentialreg} and \eqref{hlimits}, in particular, $V_\eps$ is a nonnegative bounded function.

\subsection{Spectral properties}\label{specprop}

For every $\eps > 0$, the operator $H_\eps$ is  positive, since $\inp{f}{H_\eps f} = \inp{\d_x f}{\d_x f}  + \inp{\d_y f}{\d_y f} + \inp{f}{V_\eps f} \geq 0$ holds for every $f \in H^2(\R^2)$, and we immediately obtain the following inclusion relation for the spectrum: 
$$
    \sig(H_\eps) \subseteq [0,\infty).
$$
We will show below in Theorem \ref{specthm}, under the additional condition \eqref{bcond} on $b_\eps$ and \eqref{hlimits'}  on $h_\eps$, that the spectrum of $H_\eps$ is, in fact, equal to $[0,\infty)$. Moreover, we will see that $H_\eps$ has no eigenvalues, i.e., the \emph{point spectrum} $\sig_p(H_\eps)$ is empty. Recall that the \emph{discrete spectrum} $\sig_d(T)$ of an operator $T$ consists of the isolated spectral points that are eigenvalues of finite multiplicity and the essential spectrum $\sig_{\text{ess}}(T)$ is its complement within the spectrum, thus we have the disjoint union $\sig(T) = \sig_d(T) \cup \sig_{\text{ess}}(T)$ and $\sig_d(T) \subseteq \sig_p(T)$. The \emph{essential spectrum} may include eigenvalues which are nonisolated or have infinite dimensional eigenspace, hence in general it may happen that $\sig_p(T) \cap \sig_{\text{ess}}(T) \neq \emptyset$.

 Before proceeding let us briefly recall a few more operator theoretic notions: 
 A densely defined operator $B$ is said to be \emph{relatively $A$-bounded} with respect to the densely defined operator $A$, if the domain of $B$ contains that of $A$ and there are constants $c_1, c_2 \geq 0$ such that $\| B \vphi\| \leq c_1 \|A\vphi\| + c_2 \|\vphi\|$ holds for all $\vphi$ in the domain of $A$ (cf.\ \cite[Section 13.1]{HS:96}); the infimum over $c_1$ is caled the \emph{relative bound} of $B$ with repsect to $A$; in case both operators $A$ and $B$ are closed, it suffices to require only the domain property (cf.\ \cite[Section X.2]{RS:V2}). 
 Let $A$ be a self-adjoint operator, then an operator $B$ is called \emph{relatively $A$-compact}, if the domain of $B$ contains that of $A$ and $B (A + i)^{-1}$ is a compact operator (cf.\ \cite[Section 14.1]{HS:96} or \cite[Section XIII.4, page 113]{RS:V4}); in the latter case, $B$ is also relatively $A$-bounded. We will apply these notions to $A = - \Delta$ and $B$ the multiplication operator corresponding to the potential $V_\eps$.

\begin{remark}\label{RemCompSupp} (The compactly supported case) If $V_\eps$ has compact support, which is exactly the case when $h_\eps$ has compact support, then the derivation of the spectral properties is straightforward: The corresponding multiplication operator is relatively ($-\Delta$)-compact and hence $H_\eps$ has the same essential spectrum as $-\Delta$ (cf.\ \cite[Section 14.3]{HS:96} or \cite[Theorem XIII.15]{RS:V4}). Therefore, $\sig_{\text{ess}}(H_\eps) = \sig(- \Delta) = [0, \infty)$, and since we already knew $\sig(H_\eps) \subseteq [0,\infty)$, we must have that the discrete spectrum $\sig_d(H_\eps)$ is empty and therefore $\sig(H_\eps) = \sig_{\text{ess}}(H_\eps) = [0, \infty)$.

This leaves us with the question whether there can be eigenvalues embedded in $[0,\infty)$: Since the (virial) function $(x,y) \mapsto x \d_x V_\eps(x,y) + y \d_y V_\eps(x,y)$ also has compact support, we obtain from \cite[Theorem 16.1]{HS:96} that $H_\eps$ cannot possess any strictly positive eigenvalues (alternatively, this follows also from  \cite[Corollary to Theorem XIII.58]{RS:V4}). Finally, we see that $0$ cannot be an eigenvalue either: Otherwise, there would be a nonvanishing $f \in H^2(\R^2)$ with $- \Delta f + V_\eps f = 0$ on $\R^2$; upon multiplication by $\overline{f}$, integration over $\R^2$, an integration by parts implies $0 = \inp{\d_x f}{\d_x f}  + \inp{\d_y f}{\d_y f}  + \inp{V_\eps f}{f} \geq  \inp{\d_x f}{\d_x f}  + \inp{\d_y f}{\d_y f}$ due to the nonnegativity of $V_\eps$; hence $f$ would have to be a constant function.

\end{remark}

Note that since $V_\eps$ need not have compact support and does not necessarily belong to any $L^p$ with $1 \leq p < \infty$, we cannot apply the same reasoning as in the above remark assessing spectral properties for the Hamiltonian $H_\eps = - \Delta + V_\eps$. In particular, $V_\eps$ need not be a relatively ($-\Delta$)-compact perturbation { (e.g., the image of a bounded sequence of $L^2$-orthonormal test functions with supports contained in the interior of $\supp(V_\eps)$ is not relatively compact in $L^2$)}. Nevertheless, we can prove that the spectral properties  are still analogous to those in the compactly supported case under the following technical assumptions, the first being a
slightly more specific version of requirement \eqref{hlimits}: There is a net $(a_\eps)$ of real numbers $a_\eps > 0$ with $a_\eps \to \infty$ ($\eps \to 0$) and 
\beq\tag{\ref{hlimits}$'$}\label{hlimits'}
     \text{$h_\eps \to 0$ pointwise on  $S$, and $h_\eps = a_\eps$ on $\R \setminus S$}.
\eeq 
Hence every $h_\eps$ is constant and equal to $a_\eps$ outside the slits $S$, and this constant ``barrier" tends to $\infty$.
The second condition is on $b_\eps$ and about symmetry and monotonicity:
\beq\label{bcond}
   b_\eps (-x) = b_\eps(x) \quad (x \in \R) \quad\text{and}\quad b_\eps'(x) \leq 0 \quad (x \geq 0). 
\eeq
In particular, \eqref{bcond} implies
\beq\label{bcond2}
   x b_\eps'(x) \leq 0 \quad (x \in \R). 
\eeq

\begin{theorem}\label{specthm} Let $H_\eps=-\Delta+ V_\eps$ with potential $V_\eps$ according to \eqref{potentialreg}, \eqref{hlimits'}, and \eqref{bcond}. Then we have $\sig(H_\eps) = [0,\infty)$, while the point spectrum $\sig_p(H_\eps)$ is empty, i.e.\ $H_\eps$ has no eigenvalues.
\end{theorem}

\begin{proof}  \emph{Step 1, determining the spectrum:} For arbitrary fixed $\eps > 0$ we temporarily write $a:=a_\eps$ and 
$$
    V_\eps(x,y) = \chi_1(x) (a - \chi_2(y))\quad\text{with}\quad \chi_1 := b_\eps, \chi_2(y) :=  a - h_\eps(y). 
$$    
Thus $\chi_1$ and $\chi_2$ have compact supports and we obtain
$$
    H_\eps = - \Delta + \chi_1 \otimes (a - \chi_2) = \underbrace{- \d_x^2 + a \chi_1 \otimes 1 - \d_y^2}_{A}  {- 
       \underbrace{\chi_1 \otimes \chi_2}_{B}},
$$
where $B$ is a compact perturbation of the self-adjoint operator $A$ {\small (which itself is a bounded perturbation of $- \Delta$; \cite[Theorem X.12]{RS:V2})} and therefore does not change the essential spectrum (by the classical form of Weyl's theorem \cite[Example 3 in Section XIII.4]{RS:V4}), i.e., 
$$
    \sig_{\text{ess}}(H_\eps) = \sig_{\text{ess}}(A).
$$
In addition, we observe that the operator $A$ has the form
$$
    A = S_1 \otimes I + I \otimes S_2 \quad\text{with}\quad S_1 := -\d_x^2 + a \chi_1, \; {S_2 := -\d_y^2}.
$$
By \cite[Theorem VIII.33]{RS:V1} we therefore have the following relation between the spectra
$$
   \sig(A) = \ovl{\sig(S_1) + \sig(S_2)}.
$$
Each $S_j$ ($j=1,2$) is self-adjoint and positive with domain $H^2(\R) \subseteq L^2(\R)$, thus $\sig(S_j) \subseteq [0,\infty)$. {Clearly, $\sig_{\text{ess}}(S_2) = \sig_{\text{ess}}(-\d_y^2) = [0,\infty)$.} As above, {$\chi_1$} being a compact perturbation of the one-dimensional Laplacian {$- \d_x^2$} also implies {$\sig_{\text{ess}}(S_1) = \sig_{\text{ess}}(- \d_x^2) = [0,\infty)$}. Hence we have
$$
    \sig(S_j) = \sig_{\text{ess}}(S_j) = [0,\infty) \quad (j=1,2),
$$
which then implies that
$$
    \sig(A) = [0,\infty).
$$
Since there are no isolated points in the spectrum of $A$, we also obtain
$$
      \sig_{\text{ess}}(H_\eps) = \sig_{\text{ess}}(A) = \sig(A) = [0,\infty)
$$
and positivity of $H_\eps$ yields $\sig(H_\eps) \subseteq [0,\infty)$, hence
$$
      \sig(H_\eps) = \sig_{\text{ess}}(H_\eps) = [0,\infty).
$$

\medskip

\noindent\emph{Step 2, $0$ cannot be an eigenvalue of $H_\eps$:} This follows by the same reasoning as in Remark \ref{RemCompSupp}. 

\medskip

\noindent\emph{Step 3, $H_\eps$ has no positive eigenvalues:} We will show this by suitable adaptations of the statements and proofs in \cite[Corollary to and Theorem XIII.58]{RS:V4} and noting that \cite[Theorem XIII.57]{RS:V4} is applicable to the bounded potential $V_\eps$, because the operator of multiplication by $V_\eps$ clearly is relatively $(-\Delta)$-bounded with relative bound less than $1$. 

\smallskip

\noindent\emph{Substep 3(a), preparatory remarks:} Applying the \emph{radial vector field} $r \d_r := x \d_x + y \d_y$ to $V_\eps$ yields
$$
   r \d_r V_\eps(x,y) = x b_\eps'(x) h_\eps(y) + y b_\eps(x) h_\eps'(x) 
$$
and defining $\mu_1(x,y) := - x b_\eps'(x) h_\eps(y)$, $\mu_2(x,y) := y b_\eps(x) h_\eps'(x)$ we may write
\beq\label{rdrV}
   r \d_r V_\eps = - \mu_1 + \mu_2,
\eeq
where $\mu_1 \geq 0$ (thanks to \eqref{bcond2}), $\supp \mu_1 \subseteq [-c,c] \times \R$, and $\mu_1(x,y) = - a_\eps x b_\eps'(x)$ for $y \notin S$, while $\mu_2$ has compact support contained in $[-c,c] \times S$ (thanks to \eqref{hlimits'}).

Noting that $H^2(\R^2) = \F^{-1} \{ f \in L^2(\R^2) \mid z \mapsto (1 + |z|^2) f(z) \in L^2(\R^2)\} \subseteq \F^{-1} L^1(\R^2)$ {\small (by the Cauchy-Schwarz inequality; this is also a special case of \cite[Lemma 7.9.2 and Theorem 7.9.3]{Hoermander:V1})}, we deduce that for any $\psi \in H^2(\R^2)$, the nonnegative function $|\psi|^2$ on $\R^2$ is continuous, integrable, and $\lim_{|z| \to \infty} |\psi(z)|^2 = 0$. In particular, we may apply the distribution of order $0$ given in \eqref{rdrV} to it and obtain a distributional interpretation of the $L^2$ inner product $\inp{\psi}{r \d_r V_\eps \psi}$ in terms of $\dis{\mu_2}{|\psi|^2} - \dis{\mu_1}{|\psi|^2}$. Moreover, upon introducing the scaling $V_\eps^a(z) := V_\eps(a z)$ ($z \in \R^2$, $a > 0$) and noting that $\lim_{a \to 1} \inp{\psi}{\frac{V_\eps^a - V_\eps}{a -1} \psi} = - \dis{\mu_1}{|\psi|^2} + \dis{\mu_2}{|\psi|^2}$, we may even state the following variation of the \emph{virial theorem} (cf.\ \cite[Theorem XIII.59, Equation (96)]{RS:V4}), which justifies the second equality in Equation \eqref{virial}: If $\psi \in H^2(\R^2)$ is a solution to the eigenvalue equation $- \Delta \psi + V_\eps \psi = \la \psi$, then
\beq\label{virial}
  2 \inp{\psi}{-\Delta \psi} = 2 \inp{\psi}{(\la - V_\eps) \psi} = 
  \lim_{a \to 1} \inp{\psi}{\frac{V_\eps^a - V_\eps}{a -1} \psi}
  = \left( \dis{\mu_2}{|\psi|^2}  - \dis{\mu_1}{|\psi|^2} \right) 
  =: \inp{\psi}{r \d_r V_\eps \psi}.
\eeq
Recalling now from \eqref{rdrV} that outside the compact support of $\mu_2$, we have $r \d_r V_\eps = -\mu_1$, one might hope also for an appropriate extension of the classical assertion in \cite[Corollary to Theorem XIII.58]{RS:V4}, saying that for a bounded and differentiable potential $V$, the condition of repulsiveness ($\d_r V \leq 0$) near infinity implies nonexistence of positive eigenvalues. We argue for this to be true in the remaining part of this proof by sketching the few changes required in the proof of \cite[Theorem XIII.58]{RS:V4}, where compared to the statement in \cite[Theorem XIII.58]{RS:V4} we consider $V_1 = 0$ and $V_2 = V_\eps$, having to replace the conditions (ii) and (iii) in the hypothesis of that theorem by appropriate alternative properties of $V_\eps$ along the way.

\smallskip

\noindent\emph{Substep 3(b):} Following the strategy in the proof of \cite[Theorem XIII.58]{RS:V4}, entering there at the third paragraph on page 227 (the one including Equation (89)), let us suppose that $\la > 0$ and $\psi \in H^2(\R^2)$ satisfy 
$$ 
 -\Delta \psi + V_\eps \psi = \la \psi,
$$
i.e., $\psi$ is an eigenfunction for the positive eigenvalue $\la > 0$, and define 
$$
    w(r,\om) := \sqrt{r} \, \psi(r \cos \om, r \sin \om) \quad (r > 0, \om \in [0,2\pi[).
$$
We may consider $w$ as a function $(0,\infty) \to L^2(S^1)$, where we will denote the inner product in the latter space by $(. \mid .)$. With the notation 
$$
   \tilde{f}(r,\om) := f(r \cos \om, r \sin \om)
$$ 
we may write $w(r) = \sqrt{r} \tilde{\psi}(r,.)$. Recalling $\widetilde{(\Delta \psi)} = \d_r^2 \tilde{\psi} + \frac{1}{r} \d_r \tilde{\psi} + \frac{1}{r^2} \d_\om^2 \tilde{\psi}$, an elementary calculation shows that the eigenvalue equation implies (denoting $w' = \d_r w$)
$$
    w''(r) + \frac{1}{r^2} \d_\om^2 w(r) - \tilde{V_\eps}(r) w(r) = - \frac{1}{4 r^2} w(r) - \la w(r).
$$

\smallskip

\noindent\emph{Substep 3(c):} In analogy to the role of Equations (90) and (91) on page 227 in \cite{RS:V4}, one can show from the above relation (and the positivity of the self-adjoint operator $- \d_\om^2$ in $L^2(S^1)$) that the function 
$$
   F(r) := (w'(r) \mid w'(r) ) + \frac{1}{r^2} (w(r) \mid \d_\om^2 w(r) ) + ( w(r) \mid (\la - \tilde{V_\eps}(r)) w(r))
$$
satisfies for sufficiently large $r$ the inequality
\begin{multline}\label{Flowerbound}
   \diff{r} \left( r F(r) \right) \geq (1 - \frac{1}{4 r}) \norm{w'(r)}{L^2(S^1)}^2 + (\la - \frac{1}{4r}) \norm{w(r)}{L^2(S^1)}^2\\ 
      - ( w(r) \mid \tilde{V_\eps}(r) w(r)) - ( w(r) \mid r \tilde{V_\eps}'(r) w(r)), 
\end{multline}
where the final term should be interpreted in the sense of an $r$-parametrized distribution acting on $S^1$, namely as $( w(r) \mid r \tilde{V_\eps}'(r) w(r)) := \dis{\widetilde{r \d_r V_\eps}(r)}{|w(r)|^2} =  - \dis{\tilde{\mu_1}(r)}{|w(r)|^2}$, if $r$ is also larger than the diameter of the support of $\mu_2$; in fact, upon pull-back by polar coordinates we may write $\dis{\tilde{\mu_1}}{\tilde{f}} = \int_{\sqrt{\eps^2 + d^2}}^\infty \dis{\tilde{\mu_1}(r)}{\tilde{f}(r,.)} \, dr$ with $\tilde{\mu_1}(r) \in \D'(S^1)$ being a regular distribution, i.e., acting by integration with a locally integrable function. Having settled for its meaning, we see immediatly that the fourth term on the right-hand side of inequality \eqref{Flowerbound} gives the nonnegative contribution $- ( w(r) \mid r \tilde{V_\eps}'(r) w(r)) = \dis{\tilde{\mu_1}(r)}{|w(r)|^2}$ to the lower bound. Since the scalar factors of the first and second term are clearly positive for sufficiently large $r > 0$, it remains to investigate the third term $- ( w(r) \mid \tilde{V_\eps}(r) w(r)) = - \int_{S^1} r |\tilde{\psi}(r,\om)|^2 \tilde{V_\eps}(r,\om) \, d\om$, which is clearly nonpositive. However, since for large $r > 0$, $\supp(\tilde{V_\eps}(r,.))$ is contained in two small angular intervals of size proportional to $1/r$ and $\lim_{|z| \to \infty} |\psi(z)|^2 = 0$, as noted in the discussion following \eqref{rdrV} above, we conclude that $\lim_{r \to \infty} ( w(r) \mid \tilde{V_\eps}(r) w(r)) = 0$. Therefore, \eqref{Flowerbound} implies that one can find some $R > 0$ guaranteeing
$$
     \diff{r} \left( r F(r) \right) \geq \quad (r \geq R)
$$
and we may directly deduce 
$$
  r F(r) \geq R F(R) \quad (r \geq R).
$$

\smallskip

\noindent\emph{Substep 3(d):} Now from the above, the reasoning at the bottom of page 228 in \cite{RS:V4} applies to yield that $F(r) \leq 0$ for all $r \geq R$ and that it suffices to show $w(r) = 0$ for large $r$ to finish the proof. The arguments on pages 229-230 in \cite{RS:V4} for the remaining part of the proof need 
adaptation only at one point, namely on lines 11-14 near the middle of page 229, because we do not have an analogue of the inequality $- (r^2 V_2)' \geq 0$. Instead, in estimating $\diff{r}(r^2 G(m,r))$ for the analogue of the function $G(m,r)$ introduced in \cite{RS:V4} on page 229 (with the notation $w_m = r^m w$, $m \geq 0$, as in \cite{RS:V4} and $R$, $\sqrt{\la}$ in place of $R_1$, $k$ used there, respectively), namely, 
\begin{multline*}
    G(m,r) := \norm{w_m'(r)}{L^2(S^1)}^2 +(\la - \frac{\la R}{r} + \frac{m (m+1)}{r^2})  \norm{w_m(r)}{L^2(S^1)}^2\\
    + \frac{1}{r^2} (w_m(r) \mid \d_\om^2 w_m(r) - (w_m(r) \mid \tilde{V_\eps}(r) w_m(r),
\end{multline*}
we inspect a combination of two specific terms occurring upon differentiating $r^2 G(m,r)$:
\begin{multline*}
  2 r\la (1 - \frac{R}{2 r})  \norm{w_m(r)}{L^2(S^1)}^2
   - (w_m(r) \mid (r^2 \tilde{V_\eps}(r))' w_m(r))\\ 
     = r^{2m+1} 2 \la (1 - \frac{R}{2 r})  \norm{w(r)}{L^2(S^1)}^2
     - r^{2m+1} \Big( 2 ( w(r) \mid \tilde{V_\eps}(r) w(r)) + (w(r) \mid r \tilde{V_\eps}'(r)w(r)) \Big)\\
      = r^{2m +1} \Big(  2 \la (1 - \frac{R}{2 r})  \norm{w(r)}{L^2(S^1)}^2 - 2 ( w(r) \mid \tilde{V_\eps}(r) w(r)) - (w(r) \mid r \tilde{V_\eps}'(r)w(r))  \Big) 
      =: r^{2m+1} h(r)
\end{multline*}
and recognize that by the reasoning detailed above we have $h(r) \geq 0$ for sufficiently large $r > 0$. 
\end{proof}

In summary, $H_\eps$ has no eigenvalues and its continuous spectrum equals the essential spectrum given by $[0,\infty)$.
In particular, the dynamics represented by $\exp(- i t H_\eps)$ at fixed $\eps$ does not produce any bound states that would correspond to eigenvectors associated with eigenvalues from the discrete spectrum (cf.\ \cite[Section XIII.2]{RS:V4}).

\begin{remark}
The symbol of the second-order differential operator $H_\eps$ is the function on $\R^2 \times \R^2 \isom T^*(\R^2)$ given by $H_\eps(q,\theta) := |\theta|^2 + V_\eps(q) \geq |\theta|^2$, hence it is obviously uniformly elliptic. In addition, we can show  that $H_\eps$ is an \emph{operator of constant strength} on $\R^2$ in the sense of \cite[Definition 13.1.1]{Hoermander:V2}: In fact, with the weight function as defined in \cite[Example 10.1.3]{Hoermander:V2},
$$
     \widetilde{H_\eps}(q,\theta)^2 := \sum_{|\al| \leq 2} |\d_\theta^\al H_\eps(q,\theta)|^2 = \left| |\theta|^2 + |V_\eps(q)|\right|^2 + 4 |\theta|^2 + 8, 
$$ 
we obtain for arbitrary $q, r, \theta \in \R^2$ the estimate
$$
   \left( \frac{\widetilde{H_\eps}(q,\theta)}{\widetilde{H_\eps}(r,\theta)} \right)^2 \leq 
   \frac{|\theta|^4 + 4 |\theta|^2 + 8 + (2 |\theta|^2 + \norm{V_\eps}{L^\infty}) \norm{V_\eps}{L^\infty}}{|\theta|^4 + 4 |\theta|^2 + 8} \leq
   1 +  (1 +  \norm{V_\eps}{L^\infty}) \norm{V_\eps}{L^\infty}.
$$
\end{remark}

\subsection{Generalized eigenvalues}\label{geneig}

On the level of generalized functions, the family of operators $(H_\eps)_{\eps \in (0,1]}$ defines a linear map, assigning to the class represented by $(\psi_\eps)_{\eps \in (0,1]}$ the class of $(H_\eps \psi_\eps)_{\eps \in (0,1]}$. In striving for a domain as large as possible and noting that we have 
$$
   \tH := \G_{H^2(\R^2)} \subseteq \G_{L^2(\R^2)} =: \tL
$$ 
as modules over the ring of generalized complex numbers $\gC$, we thus obtain the $\gC$-linear map 
$$
   H \col \tH \to \tL, \quad \text{class of } (\psi_\eps) \mapsto \text{ class of } (H_\eps \psi_\eps).
$$

We may then study the question of eigenvalues $\la \in \gC$ of $H$ in the sense whether $H \psi = \la \psi$ can hold with some $0 \neq \psi \in \G_{H^2}$. 
From the investigation of eigenvalue problems in \cite[Section 2 and Subsection 3.2]{HKK:14}, we learn that we should 
require of a prospective eigenfunction $\psi$ to be a \emph{free vector} and not just nonzero, in order to avoid some artefacts due to the existence of zero divisors in $\gC$. Recall that a vector $\psi \in \tL$ is free, if $\psi$ is linearly independent over the ring $\gC$, i.e., $\la \psi = 0$ in $\tL$  with $\la \in \gC$ always implies $\la = 0$. A free vector certainly is nonzero, but there are nonzero vectors in $\tL$ that are not free. 

We thus define $\la \in \gC$ to be a \emph{generalized eigenvalue} of $H$, if there exists a free vector $\psi \in \tH \subseteq \tL$ such that
$$
    H \psi = \la \psi
$$
holds. In this case we call $\psi$ a \emph{generalized eigenfunction} or \emph{eigenvector} of $H$. The \emph{generalized point spectrum} $\sig_{gp}(H)$ shall be be the subset of $\gC$ consisting of the generalized eigenvalues of $H$.

Fortunately, in case of the $\gC$-module $\tL$ it is not difficult to characterize the free vectors very much along the lines of  \cite[Lemma 2.5]{HKK:14} (originating from \cite[Theorem 5.8]{Mayrhofer:08}). 
\begin{lemma}\label{frelem}
  A vector $\psi \in \tL$ is free, if and only if the generalized number $\| \psi\|_{L^2}$, represented by $(\| \psi_\eps\|_{L^2})_{\eps \in (0,1]}$, is strictly positive, i.e.,  there is some $m \in \N$ such that $\|\psi_\eps\|_{L^2} > \eps^m$ holds for small $\eps > 0$.
\end{lemma}

\begin{proof} If $\|\psi\|_{L^2}$ is strictly positive and $\la \psi = 0$ for some $\la \in \gC$, then the family $(|\la_\eps| \|\psi_\eps\|_{L^2})$ is negligible, which can only be true if $(\la_\eps)$ is negligible, thus $\la = 0$ in $\gC$.

Suppose now that $\psi$ is free. Then $\la \in \gC \setminus \{ 0\}$ implies $\la \psi \neq 0$, which in turn gives $|\la| \| \psi\|_{L^2} \neq 0$. Therefore $\|\psi\|_{L^2}$ cannot be a zero divisor in $\widetilde{\R}$ and hence is invertible by \cite[Lemma 2.2]{HKK:14}, which in turn implies the lower bound required in strict positivity.
\end{proof}

We will determine $\sig_{gp}(H)$ in Proposition \ref{gspecprop} below and find that, in particular, 
$$
  [0,\infty) \subseteq \sig_{gp}(H), 
$$
so that all spectral values found in Theorem \ref{specthm} for each $H_\eps$ are indeed generalized eigenvalues. On the other hand, it will turn out that all generalized eigenvalues of $H$ are still infinitesimally close to the subset $[0,\infty)$, which we will make precise based on the following generalized distance concept: For an arbitrary subset $D \subseteq \C$ and $\la \in \gC$ with representative $(\la_\eps)$ we define $d(\la,D) \in \widetilde{\R} \subseteq \gC$ to be the class of $(d(\la_\eps,D))_{\eps \in (0,1]}$, where as usual $d(\mu, D) := \inf\{|\mu - z| \mid z \in D\}$ for any $\mu \in \C$. 

In contrast to the classical complex number set $\C$, it may happen for a closed subset $D$ that a generalized number $\la \in \gC \setminus D$ satisfies $d(\la, D) = 0$. To give a concrete example with $D := [0,\infty)$, consider the class $\la \in \gC$ of the representative $(\la_\eps)$ with $\la_\eps := 0$ if $1/\eps \in \N$ and $\la_\eps := 1$ otherwise. Then $\la \neq r$ in $\gC$ for every $r \in [0,\infty)$, while clearly $d(\la,[0,\infty)) = 0$. 

\begin{proposition}\label{gspecprop} We have $\sig_{gp}(H) = \{ \la \in \gC \mid d(\la, [0,\infty)) = 0\}$.
\end{proposition}

\begin{proof} Let us introduce the temporary notation $E := \{ \la \in \gC \mid d(\la, [0,\infty)) = 0\}$. 

\smallskip

\noindent We show that $E \subseteq \sig_{gp}(H)$: Let $\la \in E$ with representative $(\la_\eps)$ and let $q \in \N$ be arbitrary. 

Since $(d(\la_\eps, [0,\infty))_{\eps \in (0,1]}$ is negligible, we can find $\eps_0 > 0$ and $c > 0$ such that  for all $\eps < \eps_0$ we have $d(\la_\eps, [0,\infty)) \leq c\, \eps^{q+2}$. Thus, for every such $\eps$ we can find $r_\eps \geq 0$ with $|\la_\eps - r_\eps| \leq c\, \eps^{q+1}$. 

Due to Theorem \ref{specthm} the number $r_\eps$ belongs to the spectrum of the self-adjoint operator $H_\eps$, hence is an approximate eigenvector (\cite[Proposition 9.18]{Hall:13}) and we can thus find $w_\eps \in H^2(\R^2)$  with $\|w_\eps\|_{L^2} =1$ and $\|(H_\eps - r_\eps) w_\eps\|_{L^2} \leq \eps^{q+1}$. We may therefore conclude that 
\beq\label{evest}
   \|(H_\eps-\la_\eps)w_\eps\|_{L^2} \leq \|(H_\eps - r_\eps)w_\eps\|_{L^2} + |r_\eps - \la_\eps | \|w_\eps\|_{L^2}
   \leq \eps^{q+1} + c\, \eps^{q+1} = O(\eps^q) \quad (\eps \to 0).
\eeq
Note that we further obtain $H^2$-moderateness of the family $(w_\eps)_{\eps \in (0,1]}$, since 
\begin{multline*}
   \|w_\eps\|_{H^2} \leq \|(I - \Delta) w_\eps\|_{L^2} \leq \underbrace{\|w_\eps\|_{L^2}}_{1} + \|-\Delta w_\eps\|_{L^2}
   = 1 + \| (- \Delta + V_\eps - \la_\eps) w_\eps - V_\eps w_\eps + \la_\eps w_\eps\|_{L^2}\\
   \leq 1 + \|(H_\eps - \la_\eps) w_\eps\|_{L^2} + \|V_\eps\|_{L^\infty} \|w_\eps\|_{L^2} + |\la_\eps| \|w_\eps\|_{L^2}
   = 1 + O(\eps^q) + \|V_\eps\|_{L^\infty} + |\la_\eps|.
\end{multline*}
Therefore, with $w$ denoting the class of $(w_\eps)$ in $\tH$ we have arrived at $H w = \la w$ with $\| w\|_{L^2} = 1$, thus $w$ is free and an eigenvector for $H$ with generalized eigenvalue $\la$.

\smallskip

\noindent We show that also $\sig_{gp}(H) \subseteq E$ holds by proving $\gC \setminus E \subseteq \gC \setminus \sig_{gp}(H)$: By contradiction, suppose $\la \in  \gC \setminus E$ and  $\la \in \sig_{gp}(H)$. Let $(\la_\eps)$ be a  representative of $\la$. Since $\la$ is a generalized eigenvalue, there is an element $\psi \in \tH$ with $\|\psi\|_{L^2}$ invertible in $\widetilde{\R}$ and such that $H \psi = \la \psi$. Therefore, for any representative $(\psi_\eps)$ of $\psi$, the following holds: There is some $m \in \N$ and for every $q \in \N$ some $c > 0$ such that
$$
   \|\psi_\eps\|_{L^2} \geq \eps^m \quad\text{and}\quad \| (H_\eps - \la_\eps) \psi_\eps\|_{L^2} \leq c\, \eps^q \quad (\eps \to 0).
$$

Since $\la \not\in E$ we know that $(d(\la_\eps,[0,\infty))_{\eps\in (0,1]}$ is not negligible. Hence there is some $p \in \N$ and a strictly decreasing sequence $\eps_k \to 0$ ($k \to \infty$) such that
$$
     d(\la_{\eps_k},[0,\infty)) > \eps_k^p \quad (k \in \N).
$$
In particular, $\la_{\eps_k} \in \C\setminus [0,\infty) = \C \setminus \sig(H_{\eps_k})$ and we have the standard $L^2$-operator norm estimate
$$
    \| (H_{\eps_k} - \la_{\eps_k})^{-1} \| \leq \frac{1}{d(\la_{\eps_k},\sig(H_{\eps_k}))} = \frac{1}{d(\la_{\eps_k},[0,\infty))}
$$
for the resolvent of the self-adjoint and positive operator $H_{\eps_k}$ (\cite[Exercises 9.11.17-18]{Hall:13}). The latter implies
$$
    \|(H_{\eps_k} - \la_{\eps_k}) u\|_{L^2} \geq d(\la_{\eps_k},[0,\infty)) \|u\|_{L^2} \quad \text{ for all } u \in H^2(\R^2)
$$
and hence in particular, for every $q \in \N$ and for large $k \in \N$,
$$
    c\, \eps_k^q \geq \| (H_{\eps_k} - \la_{\eps_k}) \psi_\eps\|_{L^2} \geq d(\la_{\eps_k},[0,\infty)) \|\psi_{\eps_k}\|_{L^2}
    \geq \eps_k^p \, \eps_k^m = \eps_k^{p+m},
$$
which gives a contradiction for $k \to \infty$ if $q > p+m$.
\end{proof}

\begin{remark} The classical numbers $ \la \in [0,\infty) = \sig(H_\eps) = \sig_c(H_\eps)$ inside the generalized point spectrum $\sig_{gp}(H) = \{ \la \in \gC \mid d(\la, [0,\infty)) = 0\}$ determined in Proposition \ref{gspecprop} can be interpreted also as generalized eigenvalues with generalized eigenvectors in the sense of the rigged Hilbert space $\S(\R^2) \subseteq L^2(\R^2) \subseteq \S'(\R^2)$ as described in \cite[Chapter I, Section 4]{GV4} or \cite[Chapter VIII, \pg 4.5]{DL:V3}. Namely, there exists a generalized eigenvector $F_\eps \in \S'(\R^2)$ of $H_\eps$  in the sense that
$$
    \forall \vphi \in \S(\R^2)\col\quad  \dis{F_\eps}{H_\eps \vphi} = \la \dis{F_\eps}{\vphi}.
$$
Now one might attempt to relate the family $(F_\eps)_{\eps \in (0,1]}$ to our context of generalized eigenvectors by means of regularization by mollification. However,  it seems not straightforward how to guarantee moderateness in the $H^2$ sense to yield a representative of an eigenfunction class in $\tH \subseteq \tL$.
\end{remark}

\subsection{Aspects of scattering theory}\label{scattsubs}

As already considered briefly in Remark \ref{RemCompSupp}, one may model the regularizing potentials $V_\eps$ with compactly supported functions for each $\eps$, though certainly with growing support in the $y$-direction as $\eps \to 0$ due to condition \eqref{hlimits}. A big advantage then is that standard scattering theory methods become applicable for every fixed $\eps$.

\begin{lemma}
If $V_\eps$ has compact support, then the corresponding multiplication operator on $L^2(\R^2)$ is a \emph{short range perturbation} of $P_0 := -\Delta$ in the sense of \cite[Definition 14.4.1]{Hoermander:V2}. 
\end{lemma}

\begin{proof}
In fact, the two conditions in \cite[Theorem 14.4.2]{Hoermander:V2} are easily seen to hold: Let $\Om$ denote the open unit ball in $\R^2$. 
i) For any $y \in \R^2$, the set $E_y := \{ V_\eps(.+y) u \mid u \in \D(\Omega), \norm{\Delta u}{L^2} \leq 1\}$ is bounded in $H^2_0(\Om)$ thanks to Poincar\'{e}'s inequality, thus precompact in $L^2(\Om)$ by Rellich's embedding theorem, and the continuous inclusion $L^2(\Om) \subseteq L^2(\R^2)$ gives  precompactness in $L^2(\R^2)$. ii) There is some $R > 0$ such that $\supp(V_\eps(.+ y)) \cap \Om = \emptyset$, if $y \in \R^2$ with $|y| \geq R$. Thus for large $j \in \N$ we have $V_\eps(.+y) u = 0$ for every $y \in \R^2$ with $R_{j-1} < |y| < R_j$ and may put $M_j =0$ (notation of  \cite[Theorem 14.4.2]{Hoermander:V2} and recall $R_j = 2^{j-1}$ from \cite[Equation (14.1.2)]{Hoermander:V2}), while for finitely many $j \in \N$ we may put $M_j = C \norm{V_\eps}{L^\infty}$, where $C$ is a constant taken from Poincare's inequality to obtain $\norm{V_\eps(.+y) u}{L^2} \leq \norm{V_\eps}{L^\infty} \norm{u}{L^2} \leq  \norm{V_\eps}{L^\infty} C \norm{\Delta u}{L^2}$; then clearly, $\sum_{j=1}^\infty R_j M_j < \infty$.
\end{proof}

We may combine the property of $V_\eps$ being a short range perturbation with the information from Subsection \ref{specprop} about the spectrum of $H_\eps$ being purely continuous and apply \cite[Theorems 14.4.6, 14.6.4, and 14.6.5]{Hoermander:V2}. We obtain that the \emph{wave operators} $W_\pm^\eps$, defined by
$$
     W_\pm^\eps \vphi = \lim_{t \to \pm \infty} e^{i t H_\eps} e^{i t \Delta} \vphi \qquad (\vphi \in L^2(\R))
$$
are unitary intertwiners for $H_\eps$ and $-\Delta$, in particular, the unitary group providing the solution to \eqref{ourprobleps} can be described in the form
\beq\label{wavegroup}
   e^{- i t H_\eps} = W_+^\eps e^{i t \Delta} (W_+^\eps)^{-1}.
\eeq 
Moreover, in our situation, the \emph{distorted Fourier transforms} (\cite[Definition 14.6.3]{Hoermander:V2}) yield unitary operators  $F_\pm^\eps$ on $L^2(\R^2)$ such  that both compositions $F_+^\eps \circ W_+^\eps$ and $F_-^\eps \circ W_-^\eps$ are equal to the Fourier transform $\F$ on $L^2(\R^2)$. We therefore have as an alternative to \eqref{wavegroup},
\beq\label{dFTgroup}
    F_+^\eps (e^{-i t H_\eps} \vphi)(\theta) = e^{- i t |\theta|^2} (F_+^\eps \vphi)(\theta)  \qquad (\vphi \in L^2(\R), \theta \in \R^2).
\eeq

We recall how the action of $F_+^\eps$ can be described more explicitly on the subspace $B$ of $L^2(\R^2)$ (defined in \cite[Section 14.1]{Hoermander:V2}) consisting of those $L^2$-functions $\vphi$ such that $\norm{\vphi}{B} := \sqrt{a_1(\vphi)} + \sum_{j=2}^\infty 2^{(j-1)/2} \sqrt{a_j(\vphi)} < \infty$, where $a_1(\vphi) := \int_{|x| < 1} |\vphi|^2$ and $a_j(\vphi) := \int_{2^{j-2} < |x| < 2^{j-1}} |\vphi|^2$ ($j \geq 2$); equipped with $\norm{\ }{B}$, $B$ becomes a Banach space with $\D(\R^2)$ as a dense subspace.

Let $R_0(z)$ denote the resolvent of $-\Delta$ for $z \in \C \setminus [0,\infty)$. For $\vphi \in B$ we have by \cite[Theorem 14.5.4]{Hoermander:V2}  that $z \mapsto (I + V_\eps R_0(z))^{-1} \vphi$ gives continuous maps from both $\{ z \in \C \mid \Im z \geq 0, z \neq 0\}$ and $\{ z \in \C \mid \Im z \leq 0, z \neq 0\}$ into $B$, hence for every $\vphi \in B$ and $\la > 0$, the limits
$$
    \vphi_{\la \pm i 0}^\eps := \lim_{\al \downarrow 0} (I + V_\eps R_0(\la \pm i \al))^{-1} \vphi
$$
exist in $B$, in particular, $\vphi_{\la \pm i 0}^\eps \in L^2(\R^2)$. Now $F_+^\eps \vphi$ can be described for any $\vphi \in B$ as follows (\cite[Lemma 14.6.2 and Definition 14.6.3]{Hoermander:V2}): For every $\la > 0$, we have for $\theta \in \R^2$ with $|\theta|^2 = \la$, almost everywhere in the sense of the line measure along the circle of radius $\sqrt{\la}$,
\beq\label{dFT}
    (F_+^\eps \vphi)(\theta) = \F \vphi_{\la + i 0}^\eps(\theta).
\eeq

We may apply the distorted Fourier transform in a more concrete distributional description of the solution $u_\eps$ of \eqref{ourprobleps}. There is no substantial loss for the physics of the problem to consider only initial values from $B$.
\begin{proposition} Suppose that the initial value regularizations $g_\eps$, in addition to \eqref{invalmod},  also satisfy
\beq\label{Bcond}
    \forall \eps > 0 \col \quad g_\eps \in B.
\eeq
Then we have for the distributional action of $F_+^\eps u_\eps(t)$ (at fixed time $t$) on a test function $\psi$ on $\R^2$,
\beq\label{disform}
   \dis{F_+^\eps u_\eps(t)}{\psi} =  \int_0^\infty  r e^{- i t r^2} \int_{S^1}  \F(g_{r^2 + i 0}^\eps)(r \om) \psi(r \om) \, d\om dr,
\eeq
where $d\om$ denotes the line measure on $S^1$ and we define, for every $\la > 0$, 
$$
     g_{\la + i 0}^\eps : = (I + V_\eps R_0(\la + i 0))^{-1} g_\eps.
$$
\end{proposition}

\begin{proof}
Upon applying \eqref{dFTgroup} and \eqref{dFT}, we obtain 
\begin{multline*}
   \dis{F_+^\eps u_\eps(t)}{\psi} = \dis{ e^{- i t |.|^2} F_+^\eps g_\eps }{\psi} = 
   \int_{\R^2} e^{- i t |\theta|^2} F_+^\eps g_\eps(\theta) \psi(\theta) \, d\theta\\ =
   \int_{S^1} \int_0^\infty  e^{- i t r^2} F_+^\eps g_\eps(r \om) \psi(r \om) \,r \,  d r d\om = 
    \int_{S^1} \int_0^\infty  e^{- i t r^2} \F(g_{r^2 + i 0}^\eps)(r \om) \psi(r \om) \,r \,  d r d\om\\
    =  \int_0^\infty  r e^{- i t r^2} \int_{S^1}  \F(g_{r^2 + i 0}^\eps)(r \om) \psi(r \om) \, d\om dr.
\end{multline*}
\end{proof}

\begin{remark} (On the ``boundary value'' $R_0(\la + i 0)$ of the resolvent of $-\Delta$.)
The resolvent $R_0(z) = (- \Delta - z)^{-1}$ is defined and holomorphic for $z \in \C \setminus [0,\infty)$, but the above constructions made implicit use of the expression $R_0(\la + i 0)$ for $0 < \la = \lim_{\al \downarrow 0} (\la + i \al)$ in describing the action of the distorted Fourier transform on the function space $B$ with $\D(\R^2) \subseteq B \subseteq L^2(\R^2)$ in \eqref{dFT}. Let us try to understand now $R_0(\la + i 0) f$ or its Fourier transform $\F( R_0(\la + i 0) f)$, in terms of distributions on $\R^2$ when $f$ is a test function from $\D(\R^2)$.  For $z \in \C \setminus [0,\infty)$ we certainly have
$$
      \F( R_0(z) f)(\theta)  =  \frac{\FT{f}(\theta)}{|\theta|^2 - z}  \qquad (\theta \in \R^2)
$$
and note that $\theta \mapsto 1 / (|\theta|^2 - z)$ is a smooth bounded function, hence belongs to $\S'(\R^2)$. Let $r_0(z) \in \S'(\R^2)$ be its inverse Fourier transform, then  we may write (with convolution $\S' * \S$)
$$
    R_0(z) f = r_0(z) * f \quad\text{or} \quad  \F( R_0(z) f) = \F r_0(z) \cdot \FT{f}.
$$
Inserting $z = \la + i \mu$ with $\la \in \R$,  $\mu \geq 0$ and letting $\mu \to 0+$ we would hope that the above formulae allow for a ``boundary value'' involving a distributional kernel
$$ 
    r_0(\la + i 0) = \S'\text{-}\lim_{\mu \to 0+} r_0(\la + i \mu). 
$$
We may equivalently try to define $\F r_0(\la + i 0)$ via $\lim_{\mu \to 0+} v_\mu$ in $\S'(\R^2)$, where
$$
     v_\mu (\theta) := \frac{1}{|\theta|^2 - \la - i \mu} \qquad (\theta \in \R^2, \la \in \R, \mu > 0).
$$
Let $g \in \S(\R^2)$, use polar coordinates $\theta = r \om$ with $r \geq 0$, $\om \in S^1$, and introduce $Mg(r) := \int_{S^1} g(r \om) \,d\om$ to deduce in a first step (with the change of coordinates $r^2 = s$ in the last equality)
$$
   \dis{v_\mu}{g} = \int_{\R^2} \frac{g(\theta)}{|\theta|^2 - \la - i \mu} \, d\theta = 
   \int_0^\infty \frac{r}{r^2 - \la - i\mu}  \; \int_{S^1} g(r \om) \,d\om \,dr
    = \frac{1}{2} \int_0^\infty \frac{Mg(\sqrt{s})}{s - \la - i \mu} \, ds.
$$

The easy case is $\la < 0$, in which $R_0(\la)$ is already well-defined and $1/(s - \la) = \lim_{\mu \to 0} 1/(s - \la - i \mu)$ is a bounded function on $[0,\infty)$, thus we directly obtain also
$$
   \dis{\F r_0(\la + i 0)}{g} =  \lim_{\mu \to 0+}\dis{v_\mu}{g} = 
       \int_0^\infty \int_{S^1} \frac{g(\sqrt{s}\om)}{2(s - \la)} \, d\om\,ds.
$$

The relevant case $\la \geq 0$ is tricky, since $1/(s - \la)$ has a nonintegrable singularity at $s = \la$. Upon a shift by $\la$, we would like to interpret the above integral in terms of the one-dimensional distribution $(s - i 0)^{-1}$, defined as distributional boundary value of the holomorphic function $z \mapsto 1/z$ in the lower complex half plane (\cite[Theorem 3.1.11]{Hoermander:V1}). The action of $(s - \la - i 0)^{-1}$ on a test function $\phi$ on $\R$ can be given by (cf.\ \cite[Equations (3.2.10), (3.2.10'), both on page 72]{Hoermander:V1})
$$
   \dis{(s - \la - i 0)^{-1}}{\phi} = - \int_{\R} (\log(|s - \la|) \phi'(s) \, ds + i \pi \phi(\la).
$$
We see that its suffices to have $\phi$ as $C^1$ function here.  The remaining difficulty now is that $\phi(s) := Mg(\sqrt{s})$ is so far only defined for $s \geq 0$ and will in general not be $C^1$ up to $s=0$. However, at least on the subspace of test functions $g \in \S'(\R^2)$ with $g(0) = 0$ and $\nabla g(0) = 0$, we always obtain a $C^1$ extension of $Mg$  upon setting $Mg(\sqrt{s}) := 0$ for $s < 0$ and may then write
$$
    \dis{\F r_0(\la + i 0)}{g} = \lim_{\mu \to 0+}\dis{v_\mu}{g} = \frac{1}{2} \dis{(s - \la - i 0)^{-1}}{Mg(\sqrt{s})}.
$$
\end{remark}

\section{Alternative approximations and regularizations: Cauchy problems with source term or initial value replacing the potential} \label{insopo}

We now leave the detailed specifications of the previous sections, vary certain aspects of the modeling and allow for  simplifications and approximations. Interpreting the term involving the potential, $V_\eps u_\eps$, as a source term  (right-hand side) for the equation, we investigate connections to approximate solution methods used in physics.

Assuming that we would have convergence $u_\eps \to u$ in $C(\R, \S'(\R^2))$ as $\eps \to 0$, then necessarily
$$
   g_\eps = u_\eps(0) \to u(0),
$$
and, in the sense of distributions, 
\begin{equation}\label{Vuconv}
  i V_\eps u_\eps = i \Delta u_\eps - \d_t u_\eps \to i \Delta u - \d_t u.
\end{equation}
We would then obtain that both $u_\eps$ and $V_\eps u_\eps$ converge as distributions, suggesting that  $u_\eps \to 0$ near $\{0\} \times (\R \setminus S)$  since $V_\eps \to \infty$ on that same set, where classical particles should be blocked. Thus the approximative idea to consider the set $S$ as sources of waves propagating into the region $x>0$, as it is often described in the physics literature, can be given some mathematical support.

Suppose that instead of employing truly a (regularized) potential function in the Schr\"odinger operator we look at a simplified Cauchy problem, where the influence of the ``interaction product'' $V u$ of the potential with the wave function is somehow replaced by a source term $F$ and an initial value $f$ so that we have the following inhomogeneous Cauchy problem without potential
\beq\label{WithSource}
  \d_t w = i\, \Delta w  - i F, \quad  w|_{t=0} = f. 
\eeq
The solution can be written in terms of the fundamental solution  $E \in C(\R,\S'(\R^2))$, $E(t;x,y) = \exp(i (x^2 + y^2)/ 4t)/(4 \pi i t)$, satisfying
$$
    \d_t E - i \Delta E = 0, \quad E(0) = \de, 
$$ 
with Duhamel's principle in the form
$$
   w(t) = E(t) * f - i \int_0^t E(t-\tau) * F(\tau) \, d\tau.
$$
Equivalently, upon applying a spatial Fourier transform to \eqref{WithSource} and using a notation like  $\FT{w}(t,\xi,\eta)$ etc., we obtain an ordinary differential equation with initial condition
$$
   \d_t \FT{w}  (t,\xi, \eta) = - i (|\xi|^2 + |\eta|^2) \FT{w}  (t,\xi, \eta) -  i \FT{F}(t,\xi,\eta),
   \quad \FT{w}(0,\xi,\eta) = \FT{f}(\xi,\eta).
$$
Employing the abbreviation $\theta=(\xi,\eta)$, the spatially Fourier transformed solution is given by
\beq\label{solformula}
     \FT{w}(t,\theta) = e^{- i t |\theta|^2} \FT{f}(\theta) 
       - i \int_0^t e^{- i (t - \tau) |\theta|^2} \, \FT{F}(\tau,\theta)\, d\tau.
\eeq 
(Checking with \cite[Theorem 7.6.1]{Hoermander:V1} also confirms that $\FT{E}(t;\xi,\eta) = \exp(- i t (\xi^2 + \eta^2))$.)

\subsection{The approximate solutions from theoretical physics}\label{thphresults}
An inspection of the discussions in \cite{Beau:12} and \cite{Manoukian:89} shows that the basic solution formulae obtained in theoretical physics can be put into the context of \eqref{WithSource} as follows: In a first step, let $w_0$ denote the solution to \eqref{WithSource} with initial value $f = \de_{(-x_0,0)}$ and $F = 0$, where $x_0 > 0$; thus, $w_0$ corresponds to the wave function of a free particle emitted at time $t=0$ at the location $(-x_0,0)$ and is given by 
$w_0(t,x,y) = (E(t) * \de_{(-x_0,0)})(x,y) = E(t,x+x_0,y)$. 
Suppose that $y \mapsto b_0(y)$ describes the pattern and shapes of slits in the plane $x=0$, but contrary to the potential function above now with value $1$ for passing through and $0$ for blocking, e.g., \cite{Beau:12} uses the characteristic function of one or two bounded intervals, and that the particle passes through $x=0$ at time $t_0 > 0$. 

In a second step, we consider now the solution $w_1$ to a Cauchy problem for the free Schr\"odinger equation with initial value  corresponding to the particle represented by $w_0$ passing through the slits at time $t_0$, namely 
\beq\label{u1Problem}
  \d_t w_1 = i\, \Delta w_1, \quad  w_1(t_0,x,y)=  \de_0(x) w_0(t_0,0,y) b_0(y)  
      =  \de_0(x) E(t_0,x_0,y) b_0(y) =: f_0(x,y). 
\eeq  
The qualitative properties of the intensity distribution $y \mapsto |w_1(t_0 + T, x_1, y)|^2$ found on a screen located at distance $x_1 > 0$ from the slits and at time $t_0 + T$, $T > 0$, are then studied in detail in \cite{Beau:12, Manoukian:89} for appropriate asymptotic relations between the relevant parameters from physics (de Broglie wavelength of the particle, $t_0$, and $T$) and geometry ($x_0$,  $x_1$ and shape of the slits) and seem to give reasonable approximations to the interference features seen in actual experiments. 

We can easily obtain an explicit expression from \eqref{u1Problem} upon applying the above solution formulae to $w_1(t+t_0)$, i.e., $w_1(t) = E(t-t_0) * f_0$, noting that the $x$-convolution is trivial due to the factor $\de_0(x)$ in $f_0$, i.e., $w_1(t,x,y) = \Big(E(t-t_0,x,.) * \big(E(t_0,x_0,.) b_0(.)\big)\Big)(y)$, and writing out the remaining $y$-convolution as an integral:
\begin{multline*}
   w_1(t,x,y) = \int_{-\infty}^\infty E(t-t_0,x,y-s) E(t_0,x_0,s) b_0(s) \, ds\\ 
   = \frac{1}{(4 \pi i)^2 (t-t_0) t_0} 
      \int_{-\infty}^\infty e^{i (x^2 + (y-s)^2)/4(t-t_0)} e^{i(x_0^2 + s^2)/4 t_0} b_0(s)\, ds \\ =
         \frac{- e^{\frac{i}{4}(\frac{x^2 + y^2}{t-t_0} + \frac{x_0^2}{t_0})}}{16 \pi^2 (t-t_0) t_0}
       \int_{-\infty}^\infty e^{i (s^2 - 2 ys )/4(t-t_0)} e^{i s^2/4 t_0} b_0(s)\, ds\\ = 
       \frac{- e^{\frac{i}{4}(\frac{x^2 + y^2}{t-t_0} + \frac{x_0^2}{t_0})}}{16 \pi^2 (t-t_0) t_0}
       \int_{-\infty}^\infty e^{\frac{-i s y}{2(t-t_0)}} \underbrace{e^{\frac{i t s^2}{4 t_0 (t-t_0)}}}_{\phi(t,t_0,s)} b_0(s)\, ds =
        \frac{- e^{\frac{i}{4}(\frac{x^2 + y^2}{t-t_0} + \frac{x_0^2}{t_0})}}{16 \pi^2 (t-t_0) t_0} \, 
         \F\Big( \phi(t,t_0,.) \, b_0(.) \Big)\big(y/2(t-t_0)\big).
\end{multline*}
Therefore, the corresponding intensity distribution as a function of $y$ is proportional to
$$
   16^2 \pi^4 T^2 t_0^2 \cdot |w_1(t_0 + T, x_1, y)|^2 = \Big| \F\Big( \phi_0 b_0 \Big)\big(y/2T\big)\Big|^2,
$$
where $ \phi_0(s) := \phi(t_0 + T,t_0,s) = \exp(i (t_0 + T)s^2/t_0 T)$ and we recall that $b_0$ is the characteristic function of the slits.

\subsection{A plausibility check}\label{apc}
Let $u_{0,\eps}$ denote the solution to the free Schr\"odinger equation with initial value $g_\eps$, i.e., 
$$
   \d_t u_{0,\eps} = i \Delta u_{0,\eps}, \quad u_{0,\eps}|_{t=0} = g_\eps.
$$
Comparing with $w_0$ in Subsection \ref{thphresults} we have:
If $g_\eps \to \de_{(-x_0,0)}$ as $\eps \to 0$, then $u_{0,\eps} \to w_0$.

Consider the difference between the scattered and the free solution $w_\eps := u_\eps - u_{0,\eps}$,
which is characterized by
$$
   \d_t w_\eps = i \Delta w_\eps - i V_\eps u_\eps, \quad w_\eps|_{t=0} = 0.
$$
Interpreting this via Equation \eqref{WithSource} we have the ``source term'' $V_\eps u_\eps$ and the corresponding solution formula implies
$$
   w_\eps(t) = - i \int_0^t E(t - \tau) * (V_\eps u_\eps(\tau)) \, d \tau.
$$
Now from the set-up of $V_\eps$, in case we think of $b_\eps$ approximationg $\de_0$, we could argue for an immediate plausible approximation in the integrand by using
\beq\label{ap1}
  F_\eps(t,x,y):= V_\eps(x,y) u_\eps(t,x,y) \approx  \de_0(x) h_\eps(y) u_\eps(t,0,y).
\eeq

In view of the observations following  \eqref{Vuconv}, arguing that the ``free solution from the region $x < 0$ enters through the slits'' (represented by the set $S$) at an instance of time $t_0$ according to a ``typical travel time from the source to the plane $x=0$'', and in combination with Remark \ref{remindicator} we might go a further bold step and consider the following additional ``approximate replacement'' 
\beq\label{ap2}
     \beta_\eps(t,x,y):= \delta_0(x)
     h_\eps(y) u_\eps(t,0,y) \approx 
     \delta_0(x) 1_S(y) u_{0,\eps}(t,0,y) \de_{t_0}(t).
\eeq

Combining \eqref{ap1} and \eqref{ap2} we would be using (compare with Proposition \ref{betaprop} below)
\beq\label{VuApprox}
   F_\eps(t,x,y) \approx \beta_\eps(t,x,y) \approx \de_{t_0}(t) \de_0(x)  1_S(y) u_{0,\eps}(t_0,0,y),
\eeq
which yields the following simplified approximate solution formula 
\beq\label{wApprox}
   w_\eps(t,x,y) \approx -i \Big(E(t - t_0,x, .) *  \big( u_{0,\eps}(t_0,0,.) 1_S(.)\big)\Big) (y).
\eeq

On the other hand, if we consider $w_1$ given by \eqref{u1Problem} in Subsection \ref{thphresults} and put $\widetilde{w_1}(t) := -i H(t - t_0) w_1(t)$, then an elementary computation  shows that
$$
  \d_t \widetilde{w_1} = i \Delta \widetilde{w_1} - i \de_{t_0}(t) \de_0(x) b_0(y) w_0(t_0,0,y), 
  \quad \supp \widetilde{w_1} \subseteq [t_0,\infty) \times \R^2.
$$ 
Thus, $\widetilde{w_1}$ satisfies the Cauchy problem  \eqref{WithSource} with initial value $0$ (at time $t=0$) and source term 
$$
   F(t,x,y) = \de_{t_0}(t) \de_0(x) b_0(y) w_0(t_0,0,y),
$$
which nicely matches  \eqref{VuApprox} in the typical case where $b_0 = 1_S$ and implies 
$$
    \widetilde{w_1}(t,x,y) =  -i \Big(E(t - t_0,x, .) *  \big( w_0(t_0,0,.) 1_S(.)\big)\Big) (y),
$$
which  happens to agree with the distributional limit, as $\eps \to 0$, of the right-hand side in \eqref{wApprox}, if $g_\eps \to \de_{(-x_0,0)}$.

\subsection{Improving on the plausibility of \eqref{VuApprox}}\label{iapc} The coherence result in Corollary \ref{cor} tells us that in case of an $H^1$ initial value and a smooth bounded potential the generalized solution is associated with the solution in $C([0,T],H^1(\R^2))$ obtained from the variational method. In fact, as the proof of Corollary \ref{cor} (to be found in \cite[Corollary 3.2]{Hoermann:11}) shows, we even have convergence in this latter function space. Using this special case as a motivation, we may consider the convergence property
\beq\label{con}
  u_\eps \to u \quad (\eps \to 0)\quad \text{ in } C([0,T],H^1(\R^2))
\eeq 
as a basis of a more general analysis. 
In such circumstance we have support for the approximation \eqref{ap1} under the following technical conditions for the  regularized potential $V_\eps(x,y) = b_\eps(x) h_\eps(y)$:
\beq\label{pot}
    b_\eps \geq 0, \quad \int_\R b_\eps(x) \, dx = 1, \quad \supp(b_\eps) \subseteq [-c_\eps,c_\eps], 
    \quad \norm{h_\eps}{L^\infty(\R)} \sqrt{c_\eps} \to 0 \; (\eps \to 0).
\eeq
In fact, the following proposition justifies the approximations chosen in \eqref{ap1} and \eqref{ap2} in the context of distributions. Note that in this context neither $V_\eps u_\eps$ nor $\be_\eps$ is distributionally convergent.

\begin{proposition}\label{betaprop}
 Suppose \eqref{con} and \eqref{pot} hold and let $\be_\eps(t)$ denote the distribution encountered in \eqref{ap1} and \eqref{ap2}, i.e.,
$$
    \dis{\be_\eps(t)}{\vphi} :=  \int_\R h_\eps(y) u_\eps(t,0,y) \vphi(0,y) \, dy   \qquad (\vphi \in \D(\R^2)),
$$
then we have $\lim\limits_{\eps \to 0} \big(V_\eps u_\eps(t) - \be_\eps(t)\big) = 0$ in $\D'(\R^2)$ for every $t \in [0,T]$.
\end{proposition} 
\begin{proof} For any $\vphi \in \D(\R^2)$, we have (since $\int b_\eps = 1$ and $h_\eps, b_\eps \geq 0$)
\begin{multline*}
  |\dis{V_\eps u_\eps(t) - \be_\eps(t)}{\vphi}| = |\int_\R h_\eps(y) 
     \int_\R b_\eps(x) \big( u_\eps(t,x,y) \vphi(x,y)  - u_\eps(t,0,y) \vphi(0,y) \big) \, dx \, dy|\\
  \leq \int_\R h_\eps(y) 
     \int_\R b_\eps(x)\, |\underbrace{ u_\eps(t,x,y) \vphi(x,y)  - u_\eps(t,0,y) \vphi(0,y) }_{=: \ga_\eps(t,x,y)}| \, dx \, dy,
\end{multline*}
where we may write
$$
   \ga_\eps(t,x,y) = \int_0^{x} \diff{s}(u_\eps(t,s,y) \vphi(s,y) ) \, ds =
   \int_0^x  \big( \d_x u_\eps(t,s,y) \vphi(s,y) + u_\eps(t,s,y) \d_x \vphi(s,y) \big) \, ds.
$$
Inserting this above and using the fact that $|x| \leq c_\eps$ in $\supp(b_\eps)$ we obtain 
\begin{multline*}
   |\dis{V_\eps u_\eps(t) - \be_\eps(t)}{\vphi}| \leq 
   \int_\R h_\eps(y) \int_\R b_\eps(x) \int_{-c_\eps}^{c_\eps} |\d_x u_\eps(t,s,y) \vphi(s,y) + u_\eps(t,s,y) \d_x \vphi(s,y)|
   \, ds \, dx \, dy \\
   =  \int_\R h_\eps(y) \int_{-c_\eps}^{c_\eps} |\d_x u_\eps(t,s,y) \vphi(s,y) + u_\eps(t,s,y) \d_x \vphi(s,y)|
   \, ds \, dy\\ 
   \leq \norm{h_\eps}{L^\infty(\R)} \int_{[-c_\eps,c_\eps] \times \R} \big(
   |\d_x u_\eps(t,s,y) \vphi(s,y)| + |u_\eps(t,s,y) \d_x \vphi(s,y)| \big) \, d(s,y)\\ \leq 
   \norm{h_\eps}{L^\infty(\R)} \Big( \norm{\d_x u_\eps(t)}{L^2(M_\eps)} \norm{\vphi}{L^2(M_\eps)} + 
   \norm{u_\eps(t)}{L^2(M_\eps)} \norm{\d_x \vphi}{L^2(M_\eps)} \Big)\\ \leq
   C_0 \norm{h_\eps}{L^\infty(\R)} \norm{u_\eps(t)}{H^1(M_\eps)} \norm{\vphi}{H^1(M_\eps)}, 
\end{multline*}
where we have put $M_\eps := [-c_\eps, c_\eps] \times \R$ and $C_0$ is some positive constant. If $\supp(\vphi) \subseteq [-l_\vphi,l_\vphi]^2$, then we clearly have
$$
    \norm{\vphi}{H^1(M_\eps)}^2 \leq 3 \cdot 2 l_\vphi  \cdot 2 c_\eps \cdot \norm{\vphi}{W^{1,\infty}(\R^2)}^2
$$
and this implies (with some constant $C_1$ depending on $\vphi$)
\beq\label{intermediate}
   |\dis{V_\eps u_\eps(t) - \be_\eps(t)}{\vphi}| \leq C_1\, \sqrt{c_\eps}\, \norm{h_\eps}{L^\infty(\R)} \norm{u_\eps(t)}{H^1(M_\eps)}.
\eeq
Due to \eqref{con} we have $\norm{u_\eps(t)}{H^1(\R^2)} \to \norm{u(t)}{H^1(\R^2)}$ as $\eps \to 0$, hence there is some $1 > \eps_0 > 0$ such that $\norm{u(t)}{H^1(M_\eps)} \leq \norm{u_\eps(t)}{H^1(\R^2)} \leq 2 \norm{u(t)}{H^1(\R^2)}$ for all $0 < \eps < \eps_0$, inserted into \eqref{intermediate} we finally obtain from \eqref{pot} that
$$
    |\dis{V_\eps u_\eps(t) - \be_\eps(t)}{\vphi}| \leq 2 C_1 \norm{u(t)}{H^1(\R^2)} \, \sqrt{c_\eps}\, \norm{h_\eps}{L^\infty(\R)}
    \to 0 \qquad (\eps_0 > \eps \to 0).
$$
\end{proof}

\appendix
\section{Generalized function solutions and coherence properties} 

In this appendix, we review the main results of \cite{Hoermann:11}. Before going into details, we recall a few basics from the theory of Colombeau generalized functions. This can be seen as an extension of the so-called \emph{sequential approach to distributions}, where each distribution is represented by approximating regularizing weakly convergent sequences modulo null sequences. 

The fundamental idea of Colombeau-type regularization methods is to model nonsmooth objects by approximating nets of smooth functions, convergent or not, but with \emph{moderate} asymptotics and to identify regularizing nets whose differences compared to the moderateness scale are \emph{negligible}. For a modern introduction to Colombeau algebras we refer to \cite{GKOS:01}. Here we will also make use of constructions and notations from \cite{Garetto:05b}, where  generalized functions based on a locally convex topological vector space $E$ are defined:
Let $E$ be a locally convex topological vector space whose topology is given by the family of seminorms $\{p_j\}_{j\in J}$. The elements of  
$$
 \mathcal{M}_E := \{(u_\eps)_\eps\in E^{(0,1]}:\, \forall j\in J\,\, \exists N\in\N\quad p_j(u_\eps)=O(\eps^{-N})\, \text{as}\, \eps\to 0\}
$$
and
$$
 \mathcal{N}_E := \{(u_\eps)_\eps\in E^{(0,1]}:\, \forall j\in J\,\, \forall q\in\N\quad p_j(u_\eps)=O(\eps^{q})\, \text{as}\, \eps\to 0\},
$$ 
are called \emph{$E$-moderate} and \emph{$E$-negligible}, respectively. With operations defined componentwise, e.g., $(u_\eps) + (v_\eps) := (u_\eps + v_\eps)$ etc., $\mathcal{N}_E$ becomes a vector subspace of $\mathcal{M}_E$.  We define the \emph{generalized functions based on $E$} as the factor space $\G_E := \mathcal{M}_E / \mathcal{N}_E$. If $E$ is a differential algebra, i.e., an associative commutative algebra possessing commuting linear maps $\d_j \col E \to E$ ($j=1,\ldots,n$) that satisfy the Leibniz rule $\d_j( f \cdot g) = (\d_j f) \cdot g + f \cdot \d_j g$,  then $\mathcal{N}_E$ is an ideal in $\mathcal{M}_E$ and $\G_E$ is  a differential algebra as well.

Particular choices of $E$ reproduce the standard Colombeau algebras of generalized functions. For example, $E=\C$ with the absolute value as norm yields the generalized complex numbers $\G_{\C} = \widetilde{\C}$; for $\Omega \subseteq \R^d$ open, $E=\Cinf(\Omega)$ with the topology of compact uniform convergence of all derivatives provides the so-called special Colombeau algebra $\G_{\Cinf(\Omega)}=\G(\Omega)$. Recall that $\Omega \mapsto \G(\Omega)$ is a fine sheaf, thus, in particular, the restriction $u|_B$ of $u\in\G(\Omega)$ to an arbitrary open subset $B \subseteq \Omega$ is well-defined and yields $u|_B \in \G(B)$. Moreover, we may embed $\D'(\Omega)$ into $\G(\Omega)$ by appropriate localization and convolution regularization. 

If $E \subseteq \D'(\Omega)$, then certain generalized functions can be projected into the space of distributions by taking  weak limits: We say that $u \in \G_E$ is \emph{associated} with $w \in \D'(\Omega)$, if $u_\eps \to w$ in $\D'(\Omega)$ as $\eps \to 0$ holds for any (hence every) representative $(u_\eps)$ of $u$. This fact is also denoted by $u \approx w$. 

Consider open strips of the form $\Omega_T = \R^n \times\, ]0,T[ \,\subseteq \R^{n+1}$ (with $T > 0$ arbitrary) and  the spaces $E = H^\infty({\Omega_T}) = \{ h \in \Cinf(\Omega_T) : \d^\al h \in L^2(\Omega_T) \; \forall \al\in \N^{n+1}\}$ with the family of (semi-)norms 
$$
  \norm{h}{H^k} = \Big( \sum_{|\al| \leq k} 
    \norm{\d^\al h}{L^2}^2\Big)^{1/2}
   \quad (k\in \N), 
$$
as well as   
$E = W^{\infty,\infty}({\Omega_T}) = \{ h \in \Cinf(\Omega_T) : \d^\al h \in L^\infty(\Omega_T) \; \forall \al\in \N^{n+1}\}$  with the family of (semi-)norms 
$$
  \norm{h}{W^{k,\infty}} = \max_{|\al| \leq k} \norm{\d^\al h}{L^\infty} \quad (k\in \N). 
$$
Clearly, $\Omega_T$  satisfies the strong local Lipschitz property \cite[Chapter IV, 4.6, p.\ 66]{Adams:75}, hence every element of $H^\infty(\Omega_T)$ and $W^{\infty,\infty}(\Omega_T)$ belongs to $\Cinf(\ovl{\Omega_T})$ by the Sobolev embedding theorem \cite[Chapter V, Theorem 5.4, Part II, p.\ 98]{Adams:75}.

We will employ the notation 
$$
  \G_2 (\R^n \times [0,T]) := \G_{H^\infty({\Omega_T})}
     \quad\text{ and }\quad
   \G_\infty (\R^n \times [0,T]) := 
   \G_{W^{\infty,\infty}({\Omega_T})}.
$$
Thus, we will represent a generalized function $u \in \G_{2}(\R^n \times [0,T])$ by a net $(u_\eps)$ with the moderateness property
$$
    \forall k \, \exists m: \quad 
    \norm{u_{\eps}}{H^k} = O(\eps^{-m}) \quad (\eps \to 0).
$$
If $(\widetilde{u_{\eps}})$ is another representative of $u$, then 
$$
    \forall k \, \forall p: \quad 
    \norm{u_{\eps} - \widetilde{u_{\eps}}}{H^k} = O(\eps^{p}) 
    \quad (\eps \to 0).
$$
Similar constructions and notations are used in case of $E = H^\infty(\R^n)$ and $E = W^{\infty,\infty}(\R^n)$. Note that by Young's inequality (\cite[Proposition 8.9.(a)]{Folland:99}) any standard convolution regularization with a scaled mollifier of Schwartz class provides embeddings $L^2 \hookrightarrow \G_{2}$ and $L^p \hookrightarrow \G_{\infty}$ ($1 \leq p \leq \infty$).  

We give below the main existence and uniqueness result for the Cauchy problem for the Schr\"odinger equation as it was stated already in (\ref{SCPDE}-\ref{SCIC}) and that we recall here: Given generalized functions $c_1, \ldots, c_n$, $V$, and $f$ on $\R^n \times [0,T]$ and $g$ a generalized function on $\R^n$, find a unique generalized function $u$ on $\R^n \times [0,T]$ solving
\begin{align*}
  \d_t u - i\, \sum_{k=1}^n   \d_{x_k} (c_k 
    \d_{x_k} u) 
    + i V u &= f\\
    u \mid_{t=0} &= g.  
\end{align*}

\begin{theorem}\label{exunthm} Let $c_k$ 
 ($k=1,\ldots,n$) and $V$ be generalized functions in $\G_{\infty}(\R^n \times [0,T])$ possessing representing nets of real-valued functions, $f$ in  $\G_{2}(\R^n \times [0,T])$, and $g$ be in $\G_{2}(\R^n)$. Suppose\\ 
 (a) $c_k$ ($k=1\ldots,n$) and $V$ are of  log-type, that is,  for some (hence every) representative $(c_{k \eps})$ of $c_k$  and $(V_\eps)$ of $V$ we have 
$\norm{\d_t {c_{k \eps}}}{L^\infty} = O(\log({1}/{\eps}))$ and  $\norm{\d_t {V_{\eps}}}{L^\infty} = O(\log({1}/{\eps}))$ as $\eps \to 0$\\
and\\ 
(b) that the positivity conditions $c_{k \eps}(x,t) \geq c_0$ for all $(x,t) \in \R^n \times [0,T]$, $\eps\in\,]0,1]$, $k=1,\ldots,n$ with some constant $c_0 > 0$ hold (hence  with $c_0 / 2$  for any other representative and small $\eps$).\\[1mm]
Then the Cauchy problem (\ref{SCPDE}-\ref{SCIC})
has a unique solution $u \in \G_{2}(\R^n \times [0,T])$. 
\end{theorem}

We note that a regularization of an arbitrary finite-order distribution which meets the log-type conditions on the coefficients $c_k$ and $V$ in the above statement is easily achieved by employing a re-scaled mollification process as described in \cite{O:89}. 

In case of smooth coefficients a simple integration by parts argument shows that any solution to the Cauchy problem obtained from the variational method as in \cite[Chapter XVIII, \pg 7, Section 1]{DL:V5}) is a solution in the sense of distributions as well. In addition, the following result from \cite{Hoermann:11} shows further coherence with the generalized function solution.

\begin{corollary}\label{cor} Let $V$ 
and $c_k$ ($k=1,\ldots,n$) belong to $C^\infty(\Omega_T) \cap L^\infty(\Omega_T)$ with bounded time derivatives of first-order,  $g_0 \in H^1(\R^n)$, and $f_0 \in C^1([0,T],L^2(\R^n))$. Let $u$ denote the unique Colombeau generalized solution to the Cauchy problem (\ref{SCPDE}-\ref{SCIC}), where $g$, $f$ denote standard embeddings of $g_0$, $f_0$, respectively. Then $u \approx w$, where $w \in C([0,T],H^1(\R^n))$ is the unique distributional solution obtained from the variational method. 
\end{corollary}

\bibliography{SDiffraction}
\bibliographystyle{abbrv}

\end{document}